\documentclass{jocg}
\usepackage{hyperref}
\usepackage{fullpage}
\usepackage{amsmath}
\usepackage{enumerate}
\usepackage{url}
\usepackage{algorithm}
\usepackage{algorithmicx}
\usepackage[noend]{algpseudocode}
\usepackage{pifont}
\usepackage{polynom}
\usepackage{subfig}
\usepackage{amsfonts}
\usepackage[all]{xy}
\usepackage{graphicx}
\usepackage{color}
\usepackage{tabularx}
\usepackage{caption}
\usepackage{amsopn}

\usepackage{amsthm}
\theoremstyle{plain}
\newtheorem{theorem}{Theorem}
\theoremstyle{definition}

\newtheorem{lemma}{Lemma}
\newtheorem{note}{Note}
\newtheorem{defn}{Definition}

\DeclareMathOperator{\Conv}{\textit{Conv}}

\DeclareMathOperator{\SD}{\textit{SD}}
\DeclareMathOperator{\SkD}{\textit{SkD}}

\DeclareMathOperator{\LD}{\textit{LD}}
\DeclareMathOperator{\SphD}{\textit{SphD}}
\DeclareMathOperator{\Sph}{\textit{Sph}}

\title{%
  \MakeUppercase{Computing the Planar $\beta$-skeleton Depth}
  \thanks{Partially supported by NSERC Canada}
}

\author{%
  David~Bremner%
  \thanks{\affil{Faculty of Computer Science, University of New Brunswick,\newline Fredericton, NB, Canada}, 
          \email{\{bremner,ra.shahsavari\}@unb.ca} \newline \textbf{This paper is submitted to Journal of Computational Geometry.}}\,
  and Rasoul~Shahsavarifar\footnotemark[2]%
}

\begin{document}
\maketitle

\begin{abstract}
\label{abst}
For $\beta \geq 1$, the \emph{$\beta$-skeleton depth} ($\SkD_\beta$) of a query point $q\in \mathbb{R}^d$ with respect to a distribution function $F$ on $\mathbb{R}^d$ is defined as the probability that $q$ is contained within the \emph{$\beta$-skeleton influence region} of a random pair of points from $F$. The $\beta$-skeleton depth of $q\in \mathbb{R}^d$ can also be defined with respect to a given data set $S\subseteq \mathbb{R}^d$. In this case, computing the $\beta$-skeleton depth is based on counting all of the $\beta$-skeleton influence regions, obtained from pairs of points in $S$, that contain $q$. The $\beta$-skeleton depth introduces a family of depth functions that contains \emph{spherical depth} and \emph{lens depth} for $\beta=1$ and $\beta=2$, respectively. The straightforward algorithm for computing the $\beta$-skeleton depth in dimension $d$ takes $O(dn^2)$. This complexity of computation is a significant advantage of using the $\beta$-skeleton depth in multivariate data analysis because unlike most other data depths, the time complexity of the $\beta$-skeleton depth grows linearly rather than exponentially in the dimension $d$. The main results of this paper include two algorithms. The first one is an optimal algorithm that takes $\Theta(n\log n)$ for computing the planar spherical depth, and the second algorithm with the time complexity of $O(n^{\frac{3}{2}+\epsilon})$ is for computing the planar $\beta$-skeleton depth, $\beta >1$. By reducing the problem of \textit{Element Uniqueness}, we prove that computing the $\beta$-skeleton depth requires $\Omega(n \log n)$ time. Some geometric properties of $\beta$-skeleton depth are also investigated in this paper. These properties indicate that \emph{simplicial depth} ($\SD$) is linearly bounded by $\beta$-skeleton depth (in particular, $\SkD_\beta\geq \frac{2}{3}SD$; $\beta\geq 1$). To illustrate this relationship, the results of some experiments on random point sets are provided. In these experiments, the bounds of $\SphD\geq 2\SD$ and $\LD \geq 1.2 \SphD$ are achieved.
\end{abstract}

\section{Introduction}
\label{intro}
The rank statistic tests play an important role in univariate non-parametric statistics. If one attempts to generalize the rank tests to the multivariate case, the problem of defining a multivariate order statistic will occur. It is not clear how to define a multivariate order or rank statistic in a meaningful way. One approach to overcome this problem is to use the notion of data depth. Data depth measures the centrality of a point in a given data set in non-parametric multivariate data analysis. In other words, it indicates how deep a point is located with respect to the data set. 
\\\\Over the last decades, various notions of data depth such as \emph{halfspace depth}~\cite{hotelling1990stability,small1990survey,tukey1975mathematics}, \emph{simplicial depth}~\cite{liu1990notion}, \emph{Oja depth}~\cite{oja1983descriptive}, \emph{regression depth}~\cite{rousseeuw1999regression}, and others have emerged as powerful tools for non-parametric multivariate data analysis. Most of them have been defined to solve specific problems in data analysis. They are different in application, definition, and geometry of their central regions (regions with the maximum depth). Regarding the planar data depth functions, some research on the algorithmic aspects of them can be found in~\cite{aloupis2001computing,aloupis2003algorithms,bremner2008output,chan2004optimal,chen2013algorithms,christmann2006regression,liu2006data,matousek1991computing,rousseeuw1999regression}.
\\\\In 2006, Elmore, Hettmansperger, and Xuan~\cite{elmore2006spherical} defined another notion of data depth named \emph{spherical depth}. It is defined as the probability that point $q$ is contained in a closed random hyperball with the diameter $\overline{x_ix_j}$, where $x_i$ and $x_j$ are two random points from a common distribution function $F$. These closed hyperballs are known as influence regions of the spherical depth function. In 2011, Liu and Modarres~\cite{liu2011lens}, modified the definition of influence region, and defined \emph{lens depth}. Each lens depth influence region is defined as the intersection of two hyperballs $B(x_i,\Vert x_i,x_j\Vert )$ and $B(x_j,\Vert x_i,x_j\Vert )$. These influence regions of spherical depth and lens depth are the multidimensional generalization of \emph{Gabriel circles} and lunes in the definition of the \emph{Gabriel Graph}~\cite{gabriel1969new} and \emph{Relative Neighbourhood Graph}~\cite{supowit1983relative}, respectively. In 2017, Yang~\cite{yang2017beta}, generalized the definition of influence region, and introduced a familly of depth functions called $\beta$-skeleton depth, indexed by a single parameter $\beta\geq 1$. The influence region of $\beta$-skeleton depth is defined to be the intersection of two hyperballs given by $B(c_i,\frac{\beta}{2} \Vert x_i,x_j\Vert )$ and $B(c_j,\frac{\beta}{2}\Vert x_i,x_j\Vert )$, where $c_i$ and $c_j$ are some combinations of $x_i$ and $x_j$. Spherical depth and lens depth can be obtained from $\beta$-skeleton depth by considering $\beta=1$ and $\beta=2$, respectively. The $\beta$-skeleton depth  has some nice properties including symmetry about the center, maximality at the centre, vanishing at infinity, and monotonicity. Depending on whether Euclidean distance or Mahalanobis distance is used to construct the influence regions, the $\beta$-skeleton depth can be orthogonally invariant or affinely invariant. All of these properties are explored in~\cite{elmore2006spherical,liu2011lens,yang2014depth,yang2017beta}.
\\\\Although we focus on the planar case here, a notable characteristic of the $\beta$-skeleton depth  is that its time complexity grows linearly in the dimension $d$ while for most other data depths the time complexity grows exponentially. To the best of our knowledge, the current best algorithm for computing the $\beta$-skeleton depth is the straightforward algorithm which takes $O(dn^2)$.
\\\\In this paper, we present an optimal algorithm for computing the spherical depth ($\beta=1$) in $\mathbb{R}^2$ that takes $\Theta(n\log n)$ time. We also introduce an $O(n^{\frac{3}{2}+\epsilon})$ algorithm for computing the planar $\beta$-skeleton depth, $\beta>1$. Furthermore, we reduce the problem of Element Uniqueness to prove that computing the $\beta$-skeleton depth ($\beta\geq 1$) of a query point requires $\Omega (n\log n)$ time. We also investigate some geometric properties of $\beta$-skeleton depth. These properties lead us to bound the simplicial depth, spherical depth, and lens depth of a point in terms of one another. Finally, some experiments are provided to illustrate the relationships between $\beta$-skeleton depth and simplicial depth.

\section{$\beta$-skeleton Depth}
\label{sec:beta}
\begin{defn}
\label{def:influence}
The $\beta$-skeleton influence region of $x_i$ and $x_j$ ($S_{\beta}(x_i, x_j)$) for $0\leq \beta < \infty $ is defined as follows:
\begin{itemize}
\item for $\beta=0$, $S_{\beta}(x_i, x_j)$ is equivalent to the line segment $\overline{x_ix_j}$.
\item for $0 < \beta < 1$, $S_{\beta}(x_i, x_j)$ is the intersection of two balls with the radius ${\Vert x_i - x_j\Vert}/{2\beta}$ such that the boundaries contain both points $x_i$ and $x_j$.
\item for $1 \leq \beta < \infty$, the lune  based version of $S_{\beta}(x_i, x_j)$ is defined as:
\begin{equation}
\label{eq:beta-influence}
S_{\beta}(x_i, x_j)= B(c_i,r) \cap B(c_j,r),
\end{equation}
where $r=\frac{\beta}{2}\Vert x_i - x_j\Vert$, $c_i=\frac{\beta}{2}x_i + (1-\frac{\beta}{2})x_j$,  and $c_j=(1-\frac{\beta}{2})x_i + \frac{\beta}{2}x_j$.  
\end{itemize}
Figure~\ref{fig:betasphericallens} illustrates the $\beta$-skeleton influence regions for different values of $\beta$.
\end{defn}

\begin{note}
It seems that for $0<\beta <1$, the $\beta$-skeleton influence region in dimension $d>2$ is not well defined because the two balls in Definition~\ref{def:influence} are not unique. Since the $\beta$-skeleton depth is defined for $1\leq \beta<\infty$~\cite{yang2017beta}, we ignore the $\beta <1$ in our study of the $\beta$-skeleton depth and its influence region.
\end{note}
  
\begin{note}
In literature, the ball based version of $S_{\beta}(x_i, x_j)$, $1 \leq \beta < \infty$ is also defined. In this case, the $S_{\beta}(x_i, x_j)$ is given by the union of the balls, instead of the intersection of them in Equation~\eqref{eq:beta-influence}. For example, the hatched area in Figure~\ref{fig:betasphericallens} denotes the ball based version of the $S_{2}(x_i, x_j)$. Since the definition of the $\beta$-skeleton depth is given based on the lune based $S_{\beta}(x_i, x_j)$ alone, by $S_{\beta}(x_i, x_j)$ we only mean its lune based version hereafter in this paper.
\end{note}
  
\begin{figure*}[!ht]
  \centering
    \includegraphics[width=0.75\textwidth]{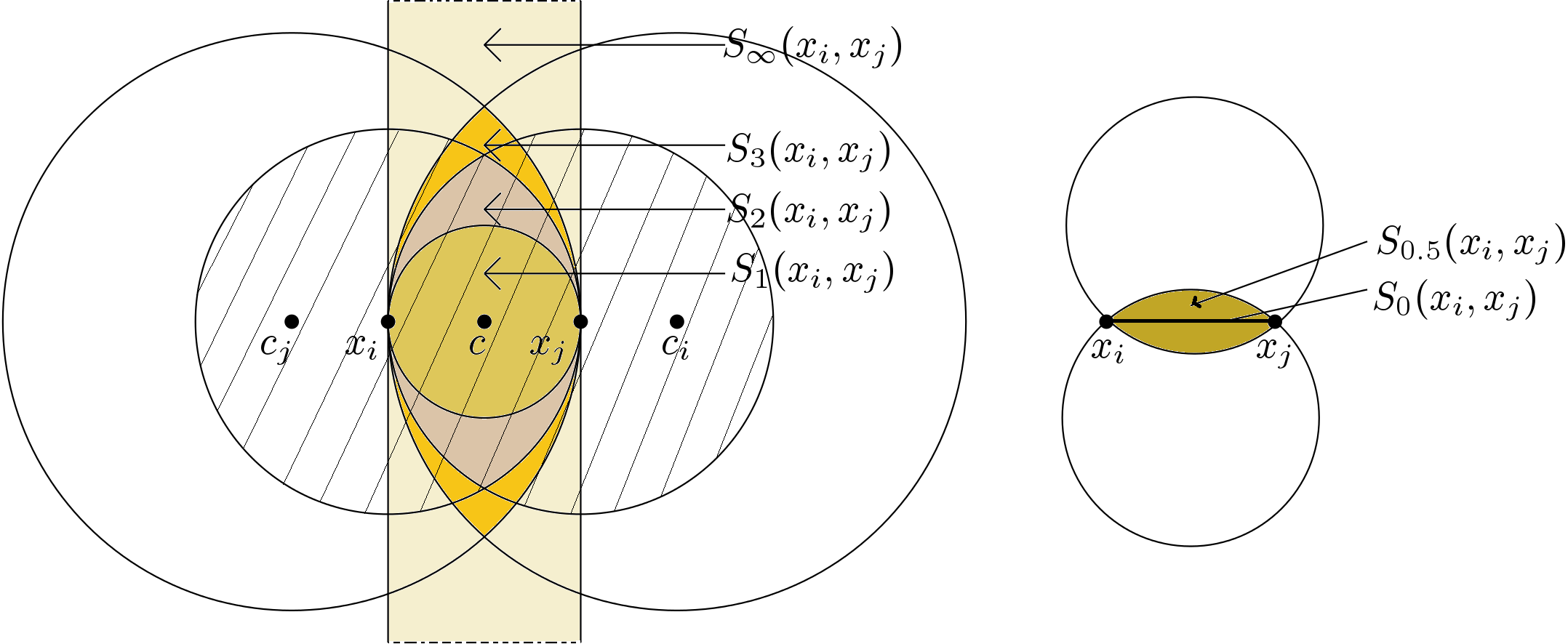}
  \caption{The $\beta$-skeleton influence regions defined by $x_i$ and $x_j$ for $\beta$=0, 0.5, 1, 2, 3, and $\beta\rightarrow\infty$, where $c=\frac{x_i+x_j}{2}$, $c_i=\frac{3}{2}x_i+(1-\frac{3}{2})x_j$, and $c_j=(1-\frac{3}{2})x_i+\frac{3}{2}x_j$}
  \label{fig:betasphericallens}
\end{figure*}

\begin{defn}
For integers $i$ and $j$ ($i<j$), we use the shorter notation $\{x_i..x_j\}$ to represent the set $\{x_i,x_{i+1
},...,x_j\}$.
\end{defn} 

\begin{defn}
For a given data set $S=\{x_1..x_n\}$ of points in general positions in $\mathbb{R}^d$ and the parameter $\beta$, the $\beta$-skeleton is defined as a graph $(S,E)$, such that $x_ix_j \in E$ if and only if no point in $S\backslash\{x_i,x_j\}$ belongs to $S_{\beta}(x_i, x_j)$, $1\leq i,j\leq n$.
\end{defn}   

\begin{defn}
The $\beta$-skeleton depth is defined as the probability that a point is contained within the $\beta$-skeleton influence region of two random vectors from a common distribution. For a distribution function $F$ on $\mathbb{R}^d$, and a vector $q$ in $\mathbb{R}^d$, the $\beta$-skeleton depth of $q$ with respect to $F$ is defined by equation~\eqref{eq:betadistribution}, where $x_i$ and $x_j$ are two random observations from $F$. 
\begin{equation}
\label{eq:betadistribution}
\SkD_{\beta}(q;F)= P[q \in S_{\beta}(x_i, x_j)];\; \beta\geq 1
\end{equation}
\end{defn}
\begin{defn}
Let $S=\{x_1..x_n\}$ be a set of points in $\mathbb{R}^d$. The $\beta$-skeleton depth of a point $q \in \mathbb{R}^d$ with respect to $S$, is defined as a proportion of the $\beta$-skeleton influence regions of $S_{\beta}(x_i, x_j), 1\leq i < j \leq n$ that contain $q$. Using the indicator function $I$, this definition can be represented by Equation~\eqref{eq:beta}.
\end{defn}
\begin{equation}
\label{eq:beta}
\SkD_{\beta}(q;S)= \frac{1}{{n \choose 2}}\sum_{1\leq i<j\leq n} {I(q \in S_{\beta}(x_i, x_j)})
\end{equation}
\\Referring to Equation~\eqref{eq:beta-influence}, it can be verified that $\{q \in S_{\beta}(x_i, x_j)\}$ is equivalent to the inequality of ${\beta \Vert x_i - x_j \Vert}/{2} \geq \max \{ \Vert q - c_i \Vert , \Vert q - c_j \Vert\}$, where $\Vert c_i - c_j \Vert=\Vert (1-\beta)(x_i-x_j) \Vert = (\beta - 1)\Vert x_i - x_j \Vert$ for $\beta \geq 1$. To compute $\SkD_{\beta}(q;S)$ in a straightforward way, it is sufficient to check this inequality for all $1\leq i,j\leq n$. As such, the computational complexity of the $\beta$-skeleton depth in $\mathbb{R}^d$ is $O(d n^2)$.
\\\\In~\cite{liu2011lens,yang2014depth,yang2017beta}, it is proved that the $\beta$-skeleton depth functions satisfy the data depth framework provided by Zuo and Serfling~\cite{zuo2000general} because these depth functions are monotonic, maximized at the center, and vanishing at infinity. The $\beta$-skeleton depth functions are also orthogonally (affinely) invariant if the Euclidean (Mahalanobis) distance is used to construct the influence regions of $\beta$-skeleton depth influence regions.

\subsection{Spherical Depth and Lens Depth}

As we discussed in Section~\ref{sec:beta}, the $\beta$-skeleton depth is a family of statistical depth functions that includes the spherical depth when $\beta=1$, and the lens depth when $\beta=2$. From the equations~\eqref{eq:beta-influence} and~\eqref{eq:beta}, the definitions of spherical depth ($\SphD$) and lens depth ($\LD$) of a query point $q$ with respect to a given data set $S$ in $\mathbb{R}^d$ are as follows:

\begin{equation}
\label{eq:spherical}
\SphD(q;S)= \frac{1}{{n \choose 2}}\sum_{1\leq i<j\leq n} {I(q \in \Sph(x_i, x_j)})
\end{equation}

\begin{equation}
\label{eq:lens}
\LD(q;S)= \frac{1}{{n \choose 2}}\sum_{1\leq i<j\leq n} {I(q \in L(x_i, x_j)}),
\end{equation}
where the influence regions $\Sph(x_i,x_j)$ and $L(x_i,x_j)$ are equal to $S_1(x_i,x_j)$ and  $S_2(x_i,x_j)$, respectively.
\\Figure~\ref{fig:3points} shows  the spherical depth and lens depth of points in the plane with respect to a set of three points $p_1, p_2$, and $p_3$ in $\mathbb{R}^2$.

\begin{figure*}[!ht]
  \centering
    \includegraphics[width=0.75\textwidth]{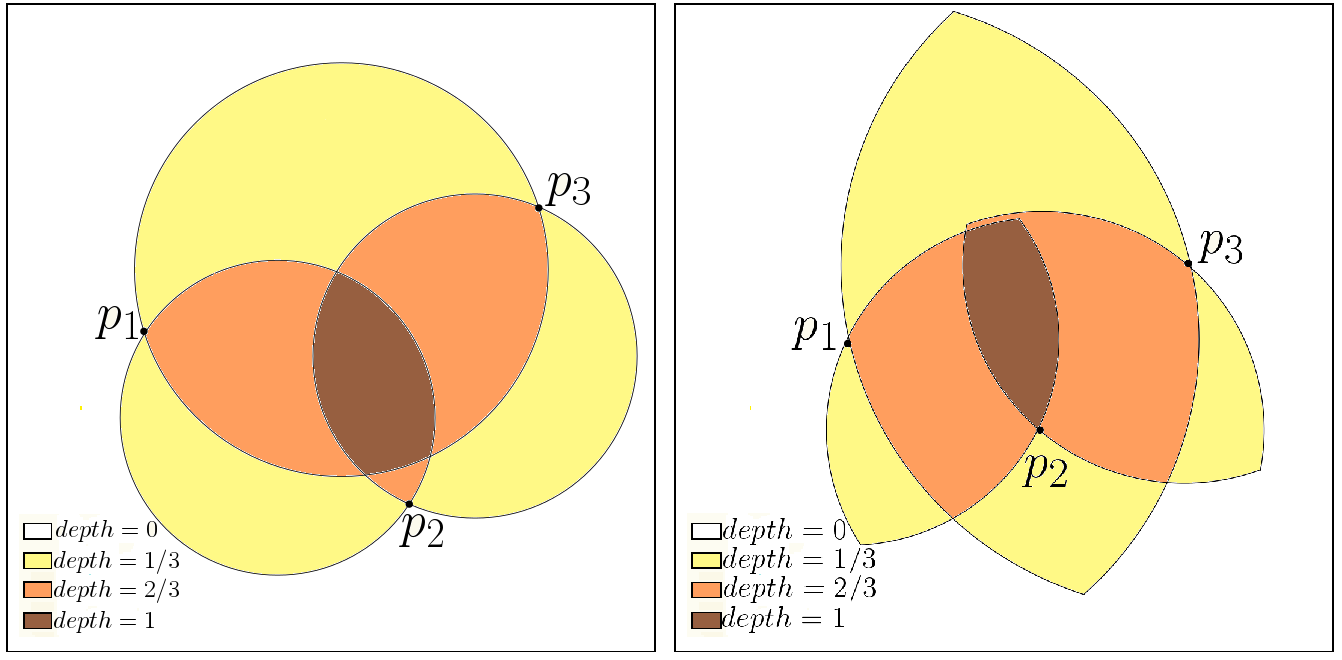}
  \caption{Spherical depth (left figure) and lens depth (right figure) of points in the plane}
  \label{fig:3points}
\end{figure*}

\section{Algorithms}

The current best algorithm for computing the $\beta$-skeleton depth of a point $q\in \mathbb{R}^d$ with respect to a data set $S=\{x_1..x_n\}\subseteq \mathbb{R}^d$ is the brute force algorithm. This naive algorithm needs to check all of the $n \choose 2$ $\beta$-skeleton influence regions obtained from the data points to figure out how many of them contain $q$. Checking all of such influence regions causes the naive algorithm to take $\Theta(dn^2)$. In this section, we present an optimal algorithm for computing the planar spherical depth ($\beta=1$) and an algorithm to compute the planar $\beta$-skeleton depth when $\beta>1$. In these algorithms, we need to solve some halfspace and some circle range counting problems, where all of the halfspaces have one common point. The circles also have the same characteristic. In the spherical depth algorithm, we have the halfspace range counting problems alone whereas in computing the $\beta$-skeleton depth, $\beta>1$ we need to solve both circle and halfspace range counting problems. 

\subsection{Optimal Algorithm for Computing the Planar Spherical Depth of a Query Point} 

Instead of checking all of the spherical  influence regions, we focus on the geometric aspects of such regions in $\mathbb{R}^2$. The geometric properties of these regions lead us to develop an $\Theta(n\log n)$ algorithm for the computation of planar spherical depth of $q\in \mathbb{R}^2$.

\begin{lemma}
For arbitrary points $a$, $b$, and $t$ in $\mathbb{R}^2$, $t \in \Sph(a,b)$ if and only if $\angle atb \geq \frac{\pi}{2}$.
\label{lm:point-circle}
\end{lemma}

\begin{proof}
If $t$ is on the boundary of $\Sph(a,b)$ \emph{Thales' Theorem}\footnote{Thales' Theorem also known as the \emph{Inscribed Angle Theorem}: If $a$, $b$, and $c$ are points on a circle where $\overline{ac}$ is a diameter of the circle, then $\angle abc$ is a right angle.} suffices as the proof in both directions. For the rest of the proof, by $t \in \Sph(a,b)$ we mean  $t\in int\Sph(a,b)$.    
\\$\Rightarrow$) For $t \in \Sph(a,b)$, suppose that $\angle atb < \frac{\pi}{2}$ (proof by contradiction). We continue the line segment $\overline{at}$ to cross the boundary of the $\Sph(a,b)$. Let $t'$ be the crossing point (see the left figure in Figure~\ref{fig:point-circle}). Since $\angle atb < \frac{\pi}{2}$, then, $\angle btt'$  is greater than $\frac{\pi}{2}$. Let $\angle btt'=\frac{\pi}{2}+\epsilon_1; \epsilon_{1}>0 $. From the Thales' Theorem, we know that $\angle at'b$ is a right angle. The angle $tbt'= \epsilon_{2}>0$ because $t\in \Sph(a,b)$. Summing up the angles in $\bigtriangleup tt'b$, as computed in~\eqref{eq:point-circle-pi-in}, leads to a contradiction. So, this direction of proof is complete.

\begin{equation}
\angle tt'b + \angle t'bt + \angle btt'\geq \frac{\pi}{2}+\epsilon_{2}+ \frac{\pi}{2}+\epsilon_{1} >\pi
\label{eq:point-circle-pi-in} 
\end{equation}
\\$\Leftarrow$) If $\angle atb = \frac{\pi}{2}+\epsilon_{1}; \epsilon_{1}>0$, we prove that $t \in \Sph(a,b)$. Suppose that $t \notin \Sph(a,b)$ (proof by contradiction). Since $t \notin \Sph(a,b)$, at least one of the line segments $\overline{at}$ and $\overline{bt}$ crosses the boundary of $\Sph(a,b)$. Without loss of generality, assume that $\overline{at}$ is the one that crosses the boundary of $\Sph(a,b)$ at the point $t'$ (see the right figure in Figure~\ref{fig:point-circle}). Considering Thales' Theorem, we know that $\angle at'b=\frac{\pi}{2}$ and consequently, $\angle bt't=\frac{\pi}{2}$. The angle $\angle t'bt=\epsilon_{2}>0$ because $t \notin \Sph(a,b)$. If we sum up the angles in the triangle $\bigtriangleup tt'b$, the same contradiction as in~\eqref{eq:point-circle-pi-in} will be implied.
\end{proof}

\begin{figure*}[!ht]
  \centering
    \includegraphics[width=0.75\textwidth]{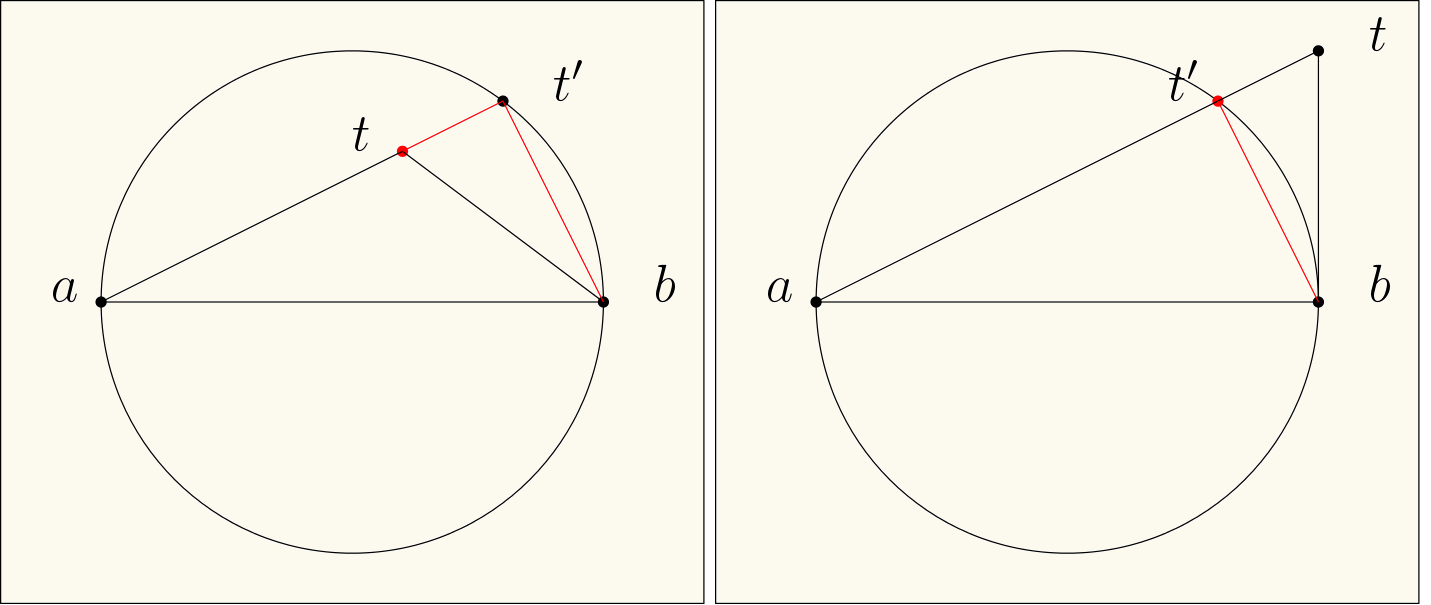}
  \caption{Point $t$ and spherical influence region $Sph(a,b)$}
  \label{fig:point-circle}
\end{figure*}

\paragraph{Algorithm~\ref{Alg:sph-pseudocode}:} Using Lemma~\ref{lm:point-circle}, we present an algorithm to compute the spherical depth of a query point $q\in \mathbb{R}^2$ with respect to $S=\{x_1..x_n\} \subseteq \mathbb{R}^2$.  This algorithm is summarized in the following steps. The pseudocode of this algorithm is provided in the Appendix. 

\begin{itemize}

\item \textbf{Translating the points:}
Suppose that $T$ is a translation by $(-q)$. We apply $T$ to translate $q$ and all data points into their new coordinates. Obviously, $T(q)= O$.

\item \textbf{Sorting the translated data points:} In this step we sort the translated data points based on their angles in their polar coordinates. After doing this step, we have $S_T$ which is a sorted array of the translated data points.

\item \textbf{Calculating the spherical depth:} For the $i^{th}$ element in $S_{T}$, we define $O_i$ and $N_i$ as follows: 

\begin{equation}
\label{eq:O_i}
\begin{split}
&O_i=\left\{j: x_j\in S_T \: , \frac{\pi}{2} \leq |\theta(x_i) -\theta(x_j)|\leq \frac{3\pi}{2}\right\} \\& N_i= \{ 1,2,...,n\} \setminus O_i.
\end{split}
\end{equation}

Thus the spherical depth of $q$ with respect to $S$, can be computed by:
\begin{equation}
\label{eq:sph}
\SphD(q;S)=\SphD(O;S_T)= \frac{1}{2}\sum_{1\leq i\leq n}|O_i|.
\end{equation}

To present a formula for computing $|O_i|$, we define $f(i)$ and $l(i)$ as follows:
\[
f(i)=
\begin{cases}
\min N_i -1 &\text{if $\frac{\pi}{2}< \theta(x_i) \leq \frac{3\pi}{2}$}\\
    \min O_i & \text{otherwise}
\end{cases}
\]
\[
l(i)=
\begin{cases}
\max N_i +1 &\text{if $\frac{\pi}{2}< \theta(x_i) \leq \frac{3\pi}{2}$}\\
    \max O_i & \text{otherwise.}
\end{cases}
\]

Figure~\ref{fig:S-T-sorted} illustrates $O_i$, $N_i$,  $f(i)$, and $l(i)$ in two different cases. Considering the definitions of $f(i)$ and $l(i)$,
\[
|O_i|=
\begin{cases}
f(i)+(n-l(i)+1) &\text{if $\frac{\pi}{2}< \theta(x_i) \leq \frac{3\pi}{2}$}\\
    l(i)-f(i)+1 & \text{otherwise.}
\end{cases}
\]
\end{itemize}
This allows us to compute $|O_i|$ using a pair of binary searches. 

\begin{figure*}[!ht]
  \centering
    \includegraphics[width=0.75\textwidth]{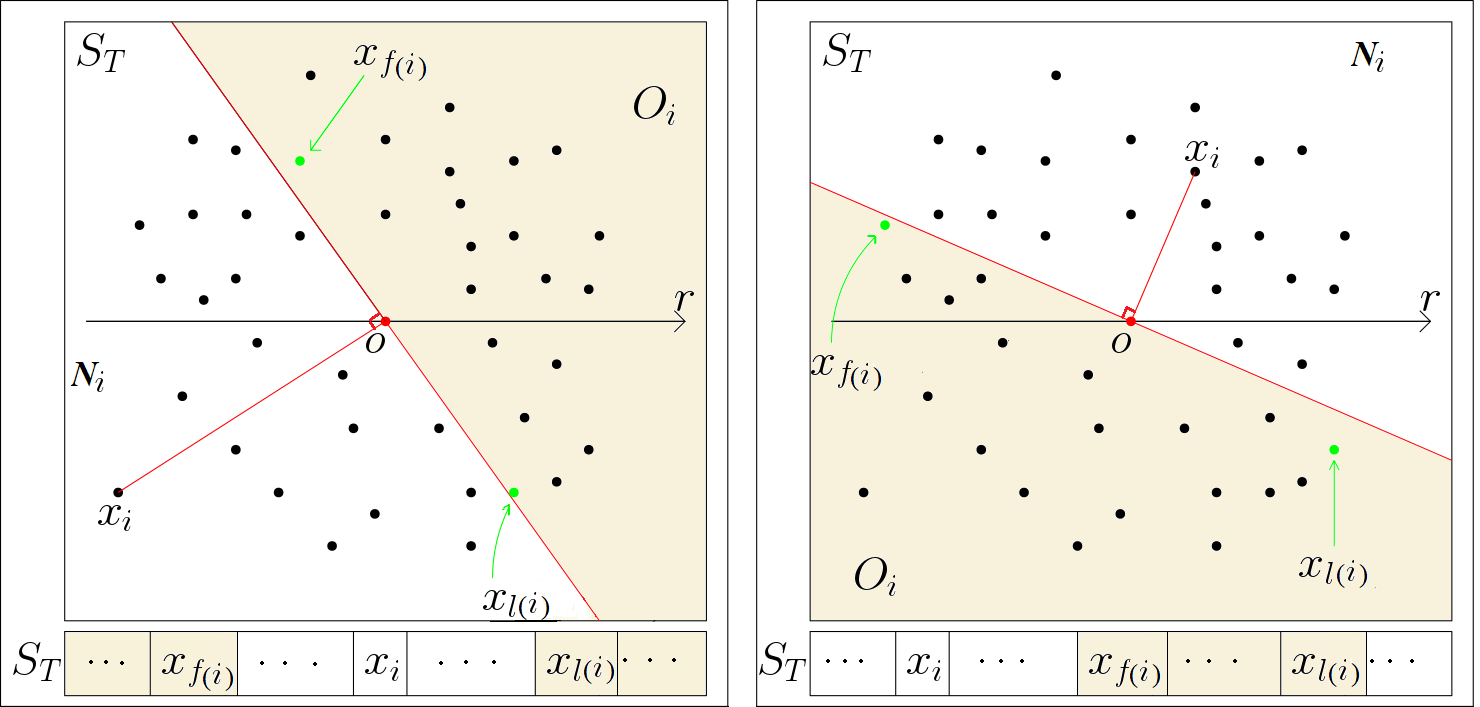}
  \caption{$\theta(x_i)\in (\frac{\pi}{2}, \frac{3\pi}{2}]$ (left figure), and $\theta(x_i)\notin (\frac{\pi}{2}, \frac{3\pi}{2}]$ (right figure)}
  \label{fig:S-T-sorted}
\end{figure*}

\paragraph{Time complexity of Algorithm~\ref{Alg:sph-pseudocode}:} The first procedure in the algorithm takes $O(n)$ to translate $q$ and all data points into the new coordinate system. The second procedure takes $O(n\log n)$ time. In this procedure, the loop iterates $n$ times, and the sorting algorithm takes $O(n\log n)$. Due to using binary search for every $O_i$, the running time of the last procedure is also $O(n\log n)$. The rest of the algorithm contributes some constant time. In total, the running time of the algorithm is $O(n\log n)$.

\paragraph{Coordinate system:} In practice it may be preferable to work in the Cartesian coordinate system. Sorting by angle can be done using some appropriate right-angle tests (determinants). Regarding the other angle comparisons, they can be done by checking the sign of dot products.

\subsection{Algorithms for Computing the Planar $\beta$-skeleton Depth ($\beta >1$ ) of a Query Point}

As illustrated in Figure~\ref{fig:betasphericallens}, $S_{\beta}(x_i,x_j) ; \beta>1$ forms some lenses, and $S_{\infty}(x_i,x_j)$ forms some slabs for different $x_i$ and $x_j$ in $S\subseteq \mathbb{R}^2$. Using some geometric properties of such lenses and slabs, we prove Lemma~\ref{lm:q-skeleton}. This lemma along with some results in range counting problems studied by Agrawal in~\cite{agarwal2013range} help us to compute $\SkD_{\beta}(q;S); \beta>1$ in $O(n^{\frac{3}{2}+\epsilon})$ time, where $q$ and $S$ are in $\mathbb{R}^2$.

\begin{defn}
\label{def:line}
For an arbitrary non-zero point $a \in \mathbb{R}^2$  and parameter $\beta \geq 1$, $\ell(p)$ is a line that is perpendicular to $\overrightarrow{a}$ at the point $p=p(a,\beta)={(\beta -1)a}/{\beta}$. This line forms two halfspaces $H_o(p)$ and $H_a(p)$. The one that includes the origin is $H_o(p)$ and the other one that includes $a$ is $H_a(p)$.
\end{defn}

\begin{defn}
\label{def:ball}
For a disk $B(c,r)$ with the center $c=c(a,\beta)={\beta a}/{2(\beta -1)}$ and radius $r=\Vert c \Vert$, $B_o(c,r)$ is the intersection of $H_o(p)$ and $B(c,r)$,  and $B_a(c,r)$ is the intersection of $H_a(p)$ and $B(c,r)$, where $\beta>1$ and $a$ is an arbitrary non-zero point in $\mathbb{R}^2$.
\\\\Figure~\ref{fig:halfspace-balls} is an illustration of these definitions for different values of parameter $\beta$.
\end{defn}

\begin{figure*}[!ht]
  \centering
    \includegraphics[width=0.75\textwidth]{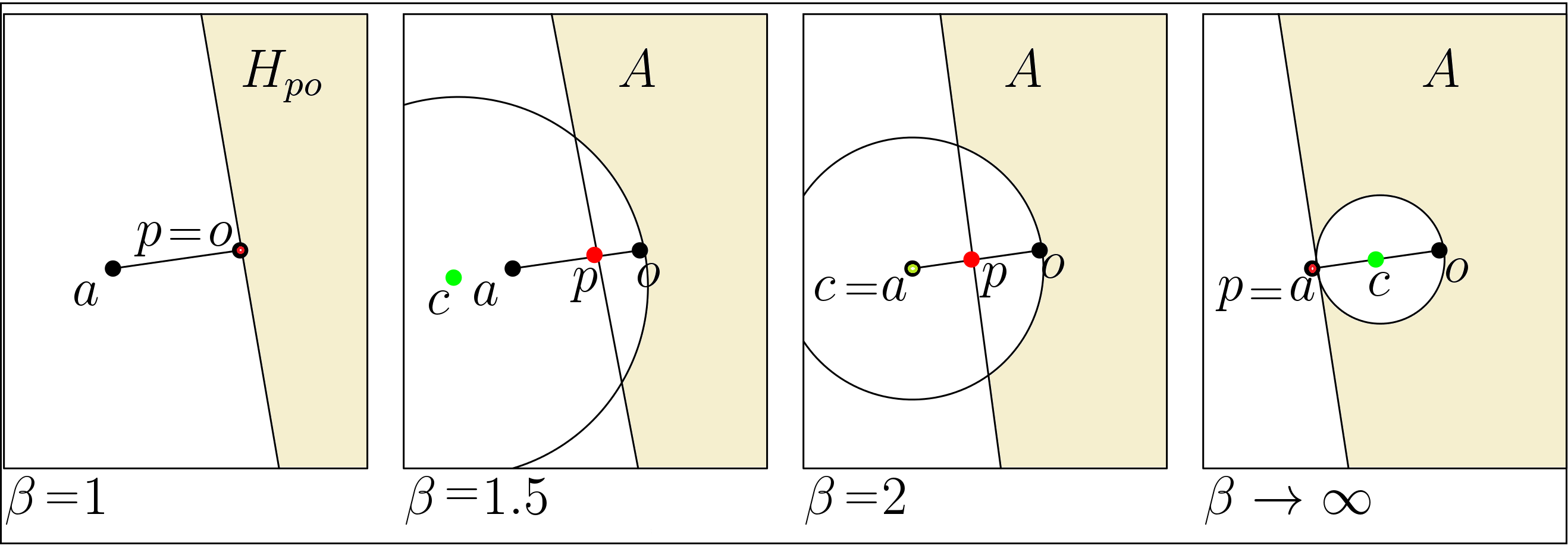}
  \caption{The $H_o(p)$ and $B(c,r)$ defined by $a\in \mathbb{R}^2$ for $\beta=1,\;1.5,\;2,\;\text{and}\;\beta\rightarrow\infty$, where $A=H_o(p)\setminus \{intB_o(c,r)\}$}
\label{fig:halfspace-balls}
\end{figure*}

\begin{lemma}
For arbitrary non-zero points $a$, $b$ in $\mathbb{R}^2$ and parameter $\beta>1$, $b \in H_o(p)\setminus \{int B_o(c,r)\}$ if and only if the origin $O=(0,0)$ is contained in $S_{\beta}(a,b)$, where $c={\beta a}/{2(\beta -1)}$, $r=\Vert c \Vert$, and $p={(\beta -1)a}/{\beta}$.
\label{lm:q-skeleton}
\end{lemma}

\begin{proof} 
First, we show that $B_o(c,r)$ is a well-defined set meaning that $\ell(p)$ intersects $B(c,r)$. We compute $d(c,\ell(p))$, the distance of $c$ from $\ell(p)$, and prove that this value is not greater than $r$. It can be verified that $d(c,\ell(p))= d(c,p)$. Let $k={\beta}/{2(\beta -1)}$; the following calculations complete this part of the proof.
\begin{align*}
d(c,p)&=d(\frac{\beta a}{2(\beta -1)},\frac{(\beta -1)a}{\beta})= d(ka,\frac{1}{2k}a)=(k-\frac{1}{2k})\sqrt{({a_x}^2+{a_y}^2)}=(\frac{2k^2-1}{2k})\Vert a \Vert\\& \leq \frac{2k^2}{2k}\Vert a \Vert= k\Vert a \Vert = r.
\end{align*}
\\\\We recall Definition~\ref{def:influence} for $\beta>1$, $S_{\beta}(a,b)=B(c_a,\frac{\beta}{2}\Vert a-b \Vert) \cap B(c_b,\frac{\beta}{2}\Vert a-b \Vert)$, where $c_a=\frac{\beta}{2}a +(1-\frac{\beta}{2})b$ and $c_b=\frac{\beta}{2}b +(1-\frac{\beta}{2})a$. Using this definition, following equivalencies can be derived from $O \in S_{\beta}(a,b)$.
\[
O \in S_{\beta}(a,b) \Leftrightarrow \frac{\beta \Vert a-b \Vert}{2} \geq max\{\Vert c_a \Vert , \Vert c_b \Vert\}\Leftrightarrow \beta \Vert a-b \Vert \geq max\{\Vert \beta(a-b)+2b\Vert , \Vert \beta(b-a)+2a \Vert\} \Leftrightarrow
\]
\[
\beta ^2 \Vert a-b \Vert ^2 \geq max\{\Vert \beta(a-b)+2b\Vert ^2 , \Vert \beta(b-a)+2a \Vert ^2\} \Leftrightarrow 0 \geq max\{b^2(1-\beta)+\beta \overrightarrow{a}.\overrightarrow{b} , a^2(1-\beta)+\beta \overrightarrow{a}.\overrightarrow{b}\}.
\]
\\By solving these inequalities for $(\beta -1)/\beta$ which is equal to $1/2k$, we have:  
\begin{equation}
\label{eq:equivalencies}
\frac{1}{2k} \geq max \left\{ \frac{\overrightarrow{a}.\overrightarrow{b}}{\Vert a \Vert ^2}, \frac{\overrightarrow{a}.\overrightarrow{b}}{\Vert b \Vert ^2} \right\}.
\end{equation}
For a fixed point $a$, the inequalities in Equation~\eqref{eq:equivalencies} determine one halfspace and one disk given by~\eqref{eq:halfspace} and~\eqref{eq:disk}, respectively.
\begin{equation}
\label{eq:halfspace}
\frac{1}{2k} \geq \frac{\overrightarrow{a}.\overrightarrow{b}}{\Vert a \Vert ^2} \Leftrightarrow \overrightarrow{a}.\overrightarrow{b} \leq\frac{1}{2k} \Vert a \Vert ^2
\end{equation}
\begin{equation}
\label{eq:disk}
\frac{1}{2k} \geq \frac{\overrightarrow{a}.\overrightarrow{b}}{\Vert b \Vert ^2} \Leftrightarrow  b^2-2k\overrightarrow{a}.\overrightarrow{b} \geq 0 \Leftrightarrow 
b^2-2k\overrightarrow{a}.\overrightarrow{b}+k^2a^2 \geq k^2a^2 \Leftrightarrow  \left( b- ka\right)^2 \geq \left( k \Vert a \Vert\right)^2
\end{equation}
The proof is complete because for a point $a$, the set of all points $b$ containing in the feasible region defined by Equations~\eqref{eq:halfspace} and~\eqref{eq:disk} is equal to $H_o(p)\setminus \{intB_o(c,r)\}$.
\end{proof}

\paragraph{Algorithm~\ref{Alg:betaskeleton-pseudocode}:} Using Lemma~\ref{lm:q-skeleton}, we present an algorithm to compute the $\beta$-skeleton depth of $q\in \mathbb{R}^2$ with respect to $S=\{x_1..x_n\} \subseteq \mathbb{R}^2$. This algorithm is summarized in two steps. Pseudocode for this algorithm can be find in the Appendix.

\begin{itemize}
\item \textbf{Translating the points:}
This step is exactly the same step as in Algorithm~\ref{Alg:sph-pseudocode}.

\item \textbf{Calculating the $\beta$-skeleton depth:} Suppose that $a=(a_x,a_y)$ is an element in $S'$ (translated $S$). We consider a disk and a line as follows:
\[
B(c,r):\left(x- ka_x \right)^2+ \left(y- ka_y\right)^2 = \left(k\Vert a \Vert\right)^2
\]
\[
\ell(p):a_xx+ a_yy= \frac{1}{2k}\Vert a \Vert,
\]
where $k$, $c$, $r$, and $p$ are defined in Lemma~\ref{lm:q-skeleton}. From Theorem 1.2 proved in~\cite{agarwal2013range}, we can compute $\vert H_o(p)\vert$ with $O(n)$ storage, $O(n\log n)$ expected preprocessing time, and $O(n^{\frac{1}{2}+\epsilon})$ query time, where $\vert H_o(p)\vert$ is the number of all elements of $S'$ that are contained in $H_o(p)$. For the elements of $S'$, $\vert intB_o(c,r) \vert$ which is defined as the number of elements containing in the interior of $B_o(c,r)$ can also be computed with the same storage, expected preprocessing time, and query time. We recall that $B_o(c,r)$ is the intersection of halfspace $H_o(p)$ and disk $B(c,r)$, where $p$, $c$, and $r$ are some functions of $a$.  Finally, $\SkD_\beta(q,S)$ which is equal to $\SkD_{\beta}(O;S')$ can be computed by Equation~\eqref{eq:skeleton-algorithm}.
 \begin{equation}
 \label{eq:skeleton-algorithm}
 \SkD_{\beta}(O;S')=\frac{1}{2}\sum_{a\in S'}(\vert H_o(p)\vert - \vert intB_o(c,r) \vert)
 \end{equation}

Note that $H_o(p)$ and $B_o(c,r)$, referring to Definitions~\ref{def:line} and~\ref{def:ball}, can be obtained from $\ell(p)$ and $B(c,r)$, respectively in constant time.   
\end{itemize}
Lemma~\ref{lm:q-skeleton} and Algorithm~\ref{Alg:betaskeleton-pseudocode} are valid for $\beta>1$. However, the case $\beta=1$ (Algorithm~\ref{Alg:sph-pseudocode} for spherical depth) can also be included if we replace $intB_o(c,r)$ with an empty set for $ \beta=1$. In this case, we need to solve the halfspace range counting problems alone and therefore, 
\[
\SphD(q;S)=\SphD(O;S')=\frac{1}{2}\sum_{a\in S'}\vert H_o(p)\vert.
\] 
\paragraph{Time complexity of Algorithm~\ref{Alg:betaskeleton-pseudocode}:} The Translating procedure as is discussed in Algorithm~\ref{Alg:sph-pseudocode}, takes $O(n)$ time. With the $O(n \log n)$ expected preprocessing time, the second procedure takes $O(n^{\frac{3}{2}+\epsilon})$ time. In this procedure, the loop iterates $n$ times, and the range counting algorithms take $O(n^{\frac{1}{2}+\epsilon})$ time. The expected preprocessing time $O(n \log n)$ is required to obtain a data structure for the aforementioned range counting algorithms. The rest of the algorithm take some constant time per loop iteration, and therefore the total expected running time of the algorithm is $O(n^{\frac{3}{2}+\epsilon})$.

\section{Lower Bounds for Computing the $\beta$-skeleton Depth of a Point in the Plane}

We reduce the problem of Element Uniqueness\footnote{ Element Uniqueness problem: Given a set $A=\{a_1..a_n\}$, is there a pair of indices $i,j$ with $i \neq j$ such that $a_i = a_j$?} to the problem of computing the $\beta$-skeleton depth. It is known that the question of Element Uniqueness has a lower bound of $\Omega (n\log n)$ in the algebraic decision tree model of computation proposed in~\cite{ben1983lower}.

\begin{theorem}
Computing the spherical depth of a query point in the plane takes $\Omega (n\log n)$ time.
\label{thrm:lowebound}
\end{theorem} 
\begin{proof}
We show that finding the spherical depth allows us to answer the question of Element Uniqueness. Suppose that $A=\{a_1..a_n\}$, for $n\geq 2$ is a given set of real numbers. We suppose all of the numbers to be positive (negative), otherwise we shift the points onto the positive X-axis. For every $a_i \in A$ we construct four points $x_i$, $x_{n+i}$, $x_{2n+i}$, and $x_{3n+i}$ in the polar coordinate system as follows:
\[
x_{(kn+i)}=\left(r_i,\theta_i+\frac{k\pi}{2}\right); \; 0\leq k \leq 3, 
\] 
where $r_i= \sqrt{1+{a_i^2}}$ and $\theta_i=\tan^{-1}(1/a_i)$. Thus we have a set $S$ of $4n$ points $x_{kn+i}$, for $1\leq i \leq n$. See Figure~\ref{fig:lower-bound-S}. The Cartesian coordinates of the points can be computed by:
\[
x_{(kn+i)}=
  \left[ {\begin{array}{cc}
   0 & -1\\
   1 & 0\\
  \end{array} } \right]^k \left( {\begin{array}{cc}
   a_i\\
   1 \\
  \end{array} } \right) ;\; k=0,1,2,3.
\]
We select the query point $q=(0,0)$, and present an equivalent form of Equation~\eqref{eq:O_i} for $O_j$ as follows:
\begin{equation}
\label{eq:O-for-lowebound}
O_j=\left\{x_k\in S\mid \angle x_jqx_k \geq \frac{\pi}{2} \right\}, \: 1\leq j\leq 4n,
\end{equation}
We compute $\SphD(q;S)$ in order to answer the Element Uniqueness problem. Suppose that every $x_j\in S$ is a unique element. In this case, $|O_j|=2n+1$ because, from~\eqref{eq:O-for-lowebound}, it can be figured out that the expanded $O_j$ is as follows:

\[
O_j=\begin{cases}
\{x_{n+1}..x_{n+j}, x_{2n+1}..x_{3n}, x_{3n+j}..x_{4n}\} ;&j\in \{1,...,n\}\\
\{x_{2n+1}..x_{n+j}, x_{3n+1}.. x_{4n}, x_{j-n}..x_{n}\} ;&j\in \{n+1,...,2n\}\\
\{x_{3n+1}..x_{n+j}, x_{1}..x_{n}, x_{j-n}..x_{2n}\} ;&j\in \{2n+1,...,3n\}\\
\{x_{1}..x_{j-3n}, x_{n+1}..x_{2n}, x_{j-n}..x_{3n}\} ;&j\in \{3n+1,...,4n\}.
\end{cases}
\]
\\\\Referring to Lemma~\ref{lm:point-circle} and Equation~\eqref{eq:sph},
\[
\SphD(q;S)=\frac{1}{2}\sum_{1\leq j\leq 4n}(2n+1)=4n^2+2n.\]
Now suppose that there exist some $i\neq j$ such that $x_i= x_j$ in $S$. In this case, from Equation ~\eqref{eq:O-for-lowebound}, it can be seen that: 
\[
|O_{(kn+i)\bmod 4n}|=|O_{(kn+j)\bmod 4n}|=2n+2,  
\]
where $k=0,1,2,3$ (see Figure~\ref{fig:lower-bound-S}). As an example, for $k=0$, $|O_j|=|O_i|=2n+2$ because the expanded form of these two sets is as follows: (without loss of generality, assume $i<j<n$)

\[
O_i=O_j=\{x_{n+1}..x_{n+j},x_{2n+1}..x_{3n},x_{3n+i},x_{3n+j}, x_{3n+j+1}..x_{4n}\}.
\]
Lemma~\ref{lm:point-circle} and Equation~\eqref{eq:sph} imply that:

\[
\SphD(q;S)\geq \frac{1}{2}(8+\sum_{1\leq j\leq 4n}(2n+1))= 4n^2+2n+4.
\]
Therefore the elements of $A$ are unique if and only if the spherical depth of $(0,0)$ with respect to $S$ is $4n^2+2n$. This implies that the computation of spherical depth requires $\Omega (n\log n)$ time. It is necessary to mention that the only computation in the reduction is the construction of $S$ which takes $O(n)$ time. At the end of the proof, we mention that the reduction does not depend on the sorted order of the elements.
\end{proof}

\begin{figure}[!ht]
  \centering
    \includegraphics[width=0.45\textwidth]{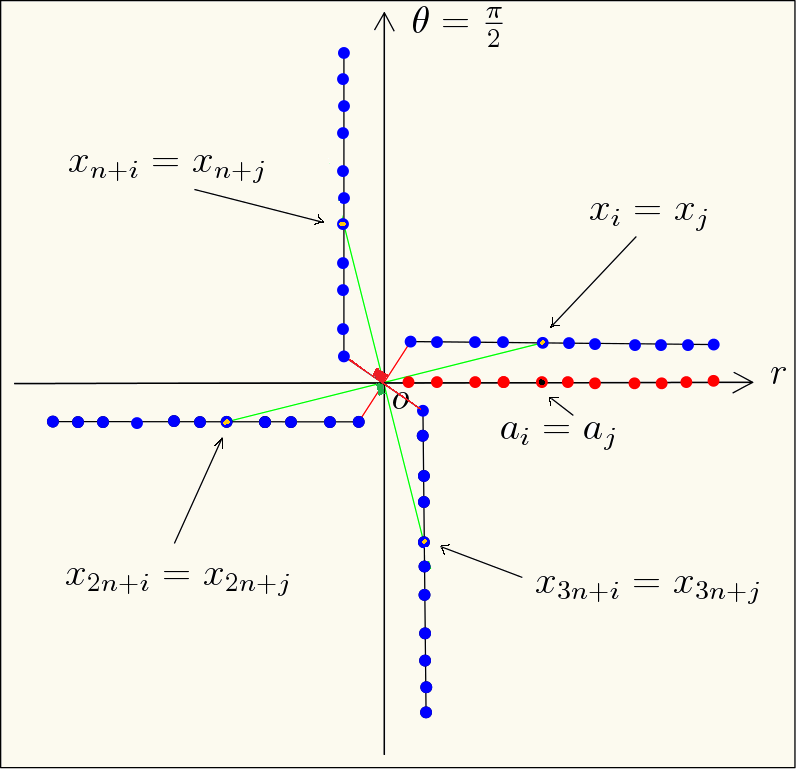}
  \caption{A representation of $A$, $S$, and duplications in these sets}
  \label{fig:lower-bound-S}
\end{figure}

\begin{note}
Instead of four copies of the elements of $A$, we could consider two copies of such elements to construct $S$. However, the depth calculation becomes more complicated in this case.
\end{note}

\begin{lemma}
Suppose that $S$ and $L_j$ are two sets defined as follows:
\[
S=\{(b_i,0),(b_i,\pi/3)\mid b_i>0, 1\leq i\leq n\}
\;\text{and}\; L_j=\{x_k\in S\mid q\in L(x_j,x_k)\}\; ,1\leq j\leq 2n.
\]
For a unique element $x_j$ in $S$, $L_j=\{x_{(n+j) \bmod 2n}\}$.
\label{lm:lens-unique}
\end{lemma}
\begin{proof}
Suppose that $L_j=\{x_k, x_{(n+j)\bmod 2n}\}$ for some $x_k \in S$ ($k\neq j$). We prove that such $x_k$ does not exist. If $\angle x_jOx_k =0$, it is obvious that $O \notin L(x_j,x_k)$ and $x_k$ cannot be an element of $L_j$. For the case $\angle x_jOx_k =\pi/3$, we assume that $O \in L(x_j,x_k)$ which means that 

\begin{equation}
\label{eq:dxjxk}
d(x_j,x_k)\geq max\{d(O,x_k),d(O,x_j)\}=max\{b_k,b_j\}.
\end{equation}
\\From the cosine formula\footnote{Cosine formula: For a triangle $\bigtriangleup abc$,
\[
\vert ab\vert ^2=\vert ac\vert ^2+\vert bc\vert ^2-2\vert ac\vert\vert bc\vert cos(\angle bca) .
\]
} in triangle $\bigtriangleup x_iOx_j$, we have

\begin{equation}
\label{eq:cosin}
d^2(x_j,x_k)= b^2_j + b^2_k -2 b_k b_j cos(\pi/3).
\end{equation}
\\Equations~\eqref{eq:dxjxk} and~\eqref{eq:cosin} imply that 
\[
d(x_j,x_k)\geq b_k \Rightarrow b^2_j -b_jb_k\geq 0 \Rightarrow b_j-b_k\geq 0 \Rightarrow b_j\geq b_k
\]
 and  
\[
d(x_j,x_k)\geq b_j\Rightarrow b^2_k -b_jb_k\geq 0 \Rightarrow b_k-b_j\geq 0 \Rightarrow b_k\geq b_j.
\]
This means that $b_k=b_j$ which contradicts the assumption of $x_k \neq x_j.$
\end{proof}

\begin{theorem}
Computing the lens depth of a query point in the plane takes $\Omega (n\log n)$ time.
\label{thrm:lowebound-L}
\end{theorem}

\begin{proof}
Suppose that $B=\{b_1..b_n\}$, for $n\geq 2$ is a given set of real numbers. Without loss of generality, we let these numbers to be positive (see the proof of Theorem~\ref{thrm:lowebound}). For $1\leq i\leq n$, we construct set $S=\{x_i, x_{n+i}\}$ of $2n$ points in the polar coordinate system such that $x_i=(b_i,0)$ and $x_{n+i}=(b_i, \pi/3)$. See Figure~\ref{fig:lower-bound-L}. We select the query point $q=(0,0)$, and define $L_j$ as follows:

\begin{equation}
\label{eq:L-for-lowebound}
L_j=\left\{x_k\in S\mid q \in L(x_j,x_k) \right\}, \: 1\leq j\leq 2n.
\end{equation}
Using Equation~\eqref{eq:L-for-lowebound}, the unnormalized form of Equation~\eqref{eq:lens} can be presented by:

\begin{equation}
\label{eq:lens-lj}
\LD_S(q)= {n \choose 2}\cdot\LD(q;S)=\frac{1}{2}\sum_{1\leq j\leq 2n}\vert L_j\vert.
\end{equation}
We solve the problem of Element Uniqueness by computing $\LD_S(q)$. Suppose that every $x_j \in S$ is a unique element. In this case, it can be verified that $L_j=\{x_{(n+j) \bmod 2n}\}$ (see Lemma~\ref{lm:lens-unique}). Equation~\eqref{eq:lens-lj} implies that:

\[
\LD_S(q)=\frac{1}{2}\sum_{1\leq j\leq 2n} 1 = n.
\]
Now assume that there exists some $i\neq j$ such that $x_i = x_j$ in $S$. In this case, 
\[
L_j=L_i=\{x_{(n+i)\bmod 2n},x_{(n+j)\bmod 2n}\} \; \text{and}\; L_{n+i}=L_{n+j}=\{x_{i},x_{j}\}.
\] 
As such,
\[
\LD_S(q)=\frac{1}{2}\sum_{1\leq j\leq 2n}\vert L_j\vert =n+2.
\]
For the case of having more duplications among the elements of $S$,
\begin{equation}
\label{eq:lens-lower}
\LD_S(q)=\frac{1}{2}\sum_{1\leq j\leq 2n}\vert L_j\vert =n+2c,
\end{equation}
where $c$ is the number of duplications. Therefore the elements of $S$ are unique if and only if $c=0$ in Equation~\eqref{eq:lens-lower}. This implies that the computation of lens depth requires $\Omega (n\log n)$. Note that all of the other computations in this reduction take $O(n)$.  
\end{proof}
\begin{note}
This reduction technique can be generalized to prove Theorem ~\ref{thrm:lowebound-beta}. It is enough to choose the rotation angle $\theta=cos^{-1}(1-1/\beta)$ in the construction of $S=\{(b_i,0),(b_i,\theta)\}$. For the case of $\beta \rightarrow \infty$, where $\theta=0$, we construct $S$ as $S=\{(b_i,0),(b_i,1)\}$. Note that we use the real \textbf{\textit{RAM}} model of computation, where we can compute the square root of a real number in constant time.
\end{note}

\begin{theorem}
For $\beta > 1$, computing the $\beta$-skeleton depth of a query point in the plane requires $\Omega(n\log n)$ time.
\label{thrm:lowebound-beta}
\end{theorem} 

\begin{figure}[!ht]
  \centering
    \includegraphics[width=0.45\textwidth]{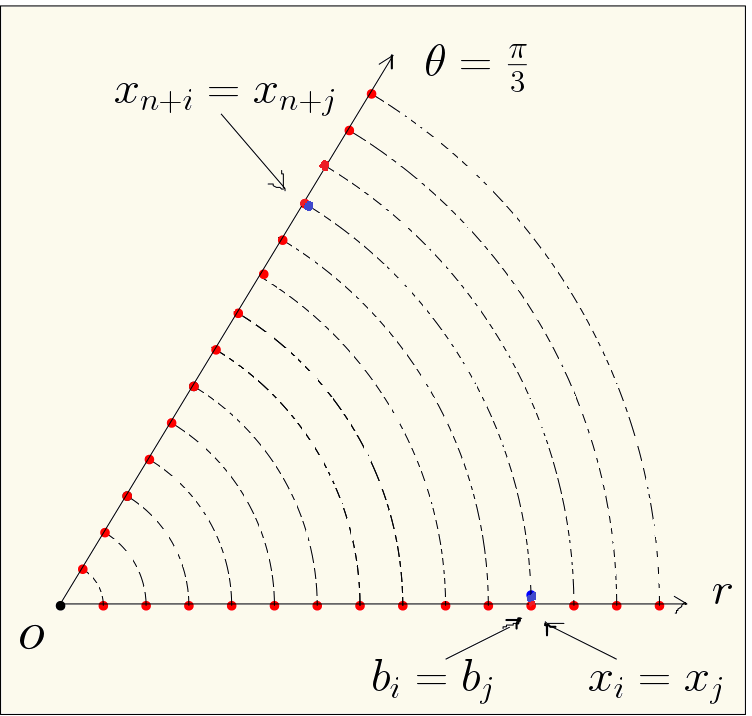}
  \caption{A representation of $B$, $S$, and duplications in these sets}
  \label{fig:lower-bound-L}
\end{figure}

\section{Relationships Among Spherical Depth, Lens Depth, and Simplicial Depth}

\begin{theorem}
For $S=\{x_1..x_n\} \subset \mathbb{R}^d$ and $q\in \mathbb{R}^d$, $\LD(q;S) \geq \SphD(q;S)$.
\label{thrm:len-sph}
\end{theorem}

\begin{proof}
From Definition~\ref{def:influence} for the spherical and lens influence regions of any arbitrary pair of points $x_i$ and $x_j$ in $S$, it can be seen that $\Sph(x_i,x_j)$ is contained in $L(x_i,x_j)$. Hence Equation~\eqref{eq:len-sph} is sufficient to complete the proof.
\begin{equation}
\label{eq:len-sph}
\SphD(q;S)= \sum _{\{x_i,x_j\} \subset S} {I(q \in \Sph(x_i, x_j))} \leq\sum _{\{x_i,x_j\} \subset S} {I(q \in L(x_i, x_j))}= \LD(q;S)     
\end{equation}
\end{proof}

\begin{defn}
\label{def:sim}
The simplicial depth of $q\in \mathbb{R}^d$ with respect to a data set $S=\{x_1..x_n\} \subset \mathbb{R}^d$ is defined by:
\begin{equation}
\label{eq:simplicial}
\SD(q;S)=\frac{1}{{n \choose d+1}}\sum _{\{x_1..x_{d+1}\}\subset S} {I(q \in \Conv\{x_1..x_{d+1}\})},
\end{equation}
where $\Conv\{x_1..x_{d+1}\}$ is a closed simplex formed by $d+1$ points of $S$.
\end{defn}

\begin{defn}
\label{def:bin-sin}
For a point $q \in \mathbb{R}^2$ and a data set $S=\{x_1..x_n\}\subset \mathbb{R}^2$, we define $B_{in}(q;S)$ to be the set of all closed spherical influence regions, out of $n \choose 2$ possible of them, that contain $q$. We also define $S_{in}(q;S)$ to be the set of all closed simplices, out of $n \choose 3$ possible defined by $S$, that contain $q$.
\end{defn}

\begin{lemma}
 Suppose that $q$ is a point in a given convex polygon $H$ obtained from a data set $S$ in $\mathbb{R}^2$. $q$ is covered by the union of spherical influence regions defined by $S$.
\label{lm:CH-q-GC}
\end{lemma}

\begin{proof}
It can be seen that there is at least one triangle, defined by the vertices of $H$, that contains $q$. We prove that the union of the spherical influence regions defined by such triangle contains $q$. See Figure~\ref{fig:Triangleabc}. We prove this statement by contradiction. Suppose that $q$ is covered by none of $\Sph(a,b)$, $\Sph(a,c)$, and $\Sph(b,c)$. Therefore, Lemma~\ref{lm:point-circle} implies that none of the angles $\angle aqb$, $\angle aqc$, and $\angle bqc$ is greater than or equal to $\frac{\pi}{2}$ which is a contradiction because at least one of these angles should be at least $\frac{2\pi}{3}$ in order to get $2\pi$ as their sum.
\end{proof}

\begin{lemma} Suppose that $S=\{a,b,c\}$ is a set of points in $\mathbb{R}^2$. For every $q\in \mathbb{R}^2$, if $\vert S_{in}(q;S)\vert =1$, then $\vert B_{in}(q;S)\vert\geq 2$.
\\\\Another form of Lemma~\ref{lm:triangle-ball} is that if $q \in \bigtriangleup abc$, then $q$ falls inside at least two spherical influence regions out of $\Sph(a,b)$, $\Sph(c,b)$, and $\Sph(a,c)$. The equivalency between these two forms of the lemma is clear. We prove just the first one.
\label{lm:triangle-ball}
\end{lemma}

\begin{proof}
From Lemma~\ref{lm:CH-q-GC}, $\vert B_{in}(q;S)\vert\geq 1$. Suppose that $\vert B_{in}(q;S)\vert =1$. If $q$ is one of the vertices of $\bigtriangleup abc$, it is clear that $\vert B_{in}(q;S)\vert \geq 2$. Without loss of generality, we suppose that $q$ falls in $int\Sph(a,b)$. For the rest of the proof, we focus on the relationships among the angles $\angle aqb$, $\angle cqa $, and $\angle cqb$ (see Figure~\ref{fig:Triangleabc}). Since $q$ is inside $\bigtriangleup abc$, $\angle aqb \leq \pi$. Consequently, at least one of $\angle cqa$ and $\angle cqb$ is greater than or equal to $\pi/2$. So, Lemma~\ref{lm:point-circle} implies that $q$ is in at least one of $int\Sph(a,c)$ and $int\Sph(b,c)$.  Hence, $\vert B_{in}(q;S)\vert =1$ contradicts $\vert S_{in}(q;S)\vert =1$ which means that $\vert B_{in}(q;S)\vert \geq 2$. As an illustration, in Figure~\ref{fig:Triangleabc}, for the points in the hatching area $\vert B_{in}(q;S)\vert =3$.
\end{proof}

\begin{figure}[!ht]
  \centering
    \includegraphics[width=0.45\textwidth]{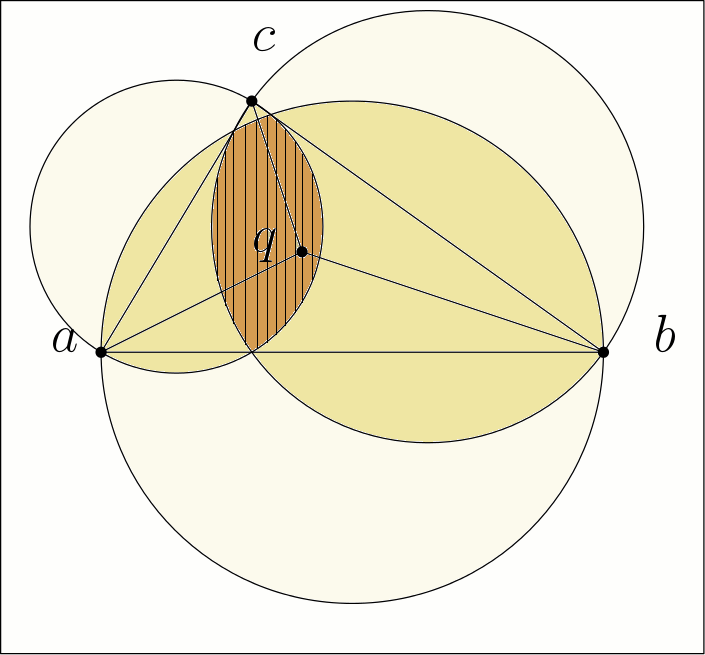}
  \caption{Triangle $abc$ contains point $q$}
  \label{fig:Triangleabc}
\end{figure}

\begin{lemma} For a data set $S=\{x_1..x_n\} \subset \mathbb{R}^2$, 
\[
\frac{\vert B_{in}(q;S)\vert }{\vert S_{in}(q;S)\vert}\geq \frac{2}{n-2}.
\]
\label{lm:bin-sin}
\end{lemma}

\begin{proof} Suppose that $\Sph(x_i,x_j)\in B_{in}(q;S)$. There exist at most $(n-2)$ triangles in $S_{in}(q;S)$ such that $\overline{x_ix_j}$ is an edge of them. We consider $ \bigtriangleup x_ix_jx_k$ to be one of such triangles (see Figure~\ref{fig:bin-sin} as an illustration). Referring to Lemma~\ref{lm:triangle-ball}, $q$ belongs to at least one of $\Sph(x_i,x_k)$ and $\Sph(x_j,x_k)$. Similarly, there exist at most $(n-2)$ triangles in $S_{in}(q;S)$ such that $x_ix_k$ (respectively $x_jx_k$) is an edge of them. In the process of computing the $\vert S_{in}(q;S)\vert$, triangle $\bigtriangleup x_ix_jx_k$ is counted at least two times, once for $\Sph(x_i,x_j)$ and another time for $\Sph(x_i,x_k)$ (or $\Sph(x_j,x_k)$ ). Consequently, for every sphere area in $B_{in}(q;S)$, there exist at most $\frac{(n-2)}{2}$ distinct triangles, triangles with only one common side, in $S_{in}(q;S)$. As a result, Equation~\eqref{eq:bin-sin} can be obtained. 

\begin{equation}
\frac{(n-2)}{2}\vert B_{in}(q;S)\vert \geq \vert S_{in}(q;S)\vert\Rightarrow \frac{\vert B_{in}(q;S)\vert}{\vert S_{in}(q;S)\vert}\geq \frac{2}{(n-2)}
\label{eq:bin-sin}
\end{equation}
\end{proof}
 
\begin{figure}[h!]
  \centering
  \includegraphics[width=0.45\textwidth]{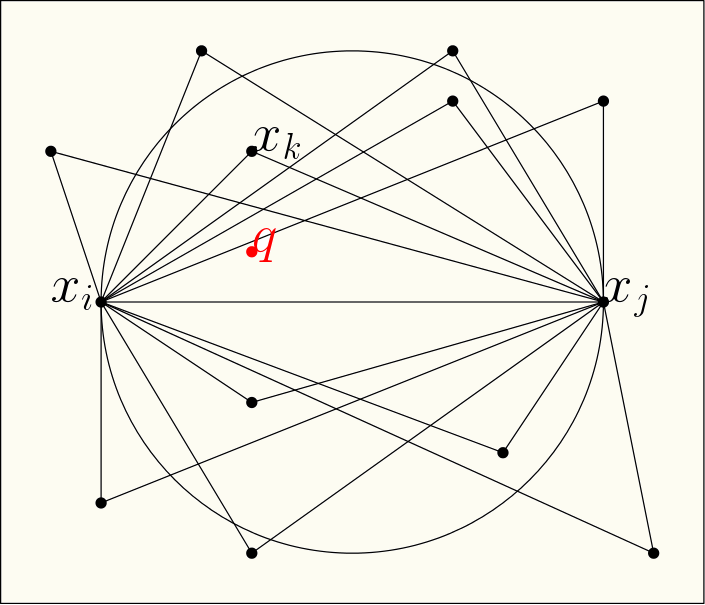}
  \caption{$q \in \Sph(x_i,x_j)$}
  \label{fig:bin-sin}
\end{figure}

\begin{theorem}
For a data set $S=\{x_1..x_n\}$ and a query point $q$ in $\mathbb{R}^2$, $\SphD(q;S)\geq \frac{2}{3} \SD(q;S)$.
\label{thrm:Sph-Simp}
\end{theorem}

\begin{proof}
From Equation~\eqref{eq:spherical} and Definitions~\ref{def:sim} and~\ref{def:bin-sin}, the ratio of spherical depth and simplicial depth can be calculated as follows: 
\begin{equation}
\label{eq:first-ratio1}
\frac{\SphD(q;S)}{\SD(q;S)}= \frac{\frac{\vert B_{in}(q;S)\vert}{{n \choose 2}}}{\frac{\vert S_{in}(q;S)\vert}{{n \choose 3}}}= \frac{(n-2)\vert B_{in}(q;S)\vert}{3\vert S_{in}(q;S)\vert}.
\end{equation}
Equation~\eqref{eq:first-ratio1} and Lemma~\ref{lm:bin-sin} imply that 
\[
\frac{\SphD(q;S)}{\SD(q;S)} \geq \frac{2}{3} \Rightarrow \SphD(q;S) \geq \frac{2}{3} \SD(q;S).
\]
\end{proof}

\section{Experiments}

To support Theorems~\ref{thrm:len-sph} and~\ref{thrm:Sph-Simp}, we compute the spherical depth, lens depth, and the simplicial depth of the points in three random sets $Q_1$, $Q_2$, and $Q_3$ with respect to data sets $S_1$, $S_2$, and $S_3$, respectively. The elements of $Q_i$ and $S_i$ are some randomly generated points (double precision floating point) within the square $\{(x,y)| x,y \in [-10,10]\}$. The results of our experiments are summarized in Table~\ref{table:results-random-points}. Every cell in the table represents the corresponding depth of $q_i$ with respect to data set $S_i$, where $q_i \in Q_i$. The cardinalities of $Q_i$ and $S_i$ are as follows: $|Q_1|=100$, $|S_1|=750$, $|Q_2|=750$, $|S_2|=2500$, $|Q_3|=2500$, $|S_3|=10000$. As can be seen in Table~\ref{table:results-random-points}, there are some gaps between obtained experimental bounds for random points and the theoretical bounds. These gaps motivate us to do more research in this area.

\begin{table}[!ht]
\begin{center}
\caption{Summary of experimental results}
\label{table:results-random-points}
\begin{tabular}{|l||l|l||l|l||l|l|}
\hline
 &\multicolumn{2}{l||}{$(q_1;S_1)$}&\multicolumn{2}{l||}{$(q_2;S_2)$}&\multicolumn{2}{l|}{$(q_3;S_3)$}\\
\cline{2-7}
 &Min& Max&Min&Max&Min&Max\\
\hline\hline
$\SD$&0.00&0.25&0.00&0.25&0.00&0.24\\
\hline
$\SphD$&0.01&0.50&0.00&0.50&0.00&0.50\\
\hline
$\LD$&0.05&0.61&0.05&0.61&0.04&0.61\\
\hline
$\frac{\SphD}{\SD}$&2.00&$\infty$&2.00&$\infty$&2.03&$\infty$\\
\hline
$\frac{\LD}{\SD}$&2.43&$\infty$&2.44&$\infty$&2.44&$\infty$\\
\hline
$\frac{\LD}{\SphD}$&1.21&8.11&1.22&23.16&1.22&157.16 \\
\hline
\end{tabular}
\end{center}
\end{table}

\section{Conclusions}

In this paper, we developed an optimal $\Theta(n\log n)$ algorithm to compute the spherical depth of a bivariate query point with respect to a given data set in $\mathbb{R}^2$. We also proposed an $O(n^{\frac{3}{2}+\epsilon})$ algorithm for computing the planar $\beta$-skeleton depth where $\beta>1$. To obtain a lower bound for computing the $\beta$-skeleton depth ($\beta\geq 1$), the Element Uniqueness problem, which requires $\Omega (n\log n)$ time, is reduced to the problem of computing the planar $\beta$-skeleton depth of $O=(0,0)$. In addition to the time complexity, the main advantage of the first algorithm is its simplicity of implementation.  We also investigated some geometric properties which lead us to find some theoretical relationships (i.e.\ $\SphD\geq \frac{2}{3} SD$ and $\LD\geq \SphD$) among spherical depth, lens depth, and simplicial depth. 
Finally, some experimental results (i.e.\ $\SphD\geq 2SD$ and $\LD\geq 1.2 \SphD$) are provided. More research on this topic is needed to figure out if the real bounds are closer to the experimental bounds or to the current theoretical bounds.

\section*{Appendix}

\begin{note}
To avoid some unusual notations in the pseudocode of Algorithm~\ref{Alg:sph-pseudocode}, we use the variables $a$ and $b$ instead of $f(i)$ and $l(i)$, respectively in the text.
\end{note}

\begin{algorithm}
\renewcommand{\algorithmicrequire}{\textbf{Input:}}
\renewcommand{\algorithmicensure}{\textbf{Output:}}
\newcommand{\Break}{\State \textbf{break}}
\caption{Computing the spherical depth of points in the plane}\label{Alg:sph-pseudocode}
\begin{algorithmic}[1]
\Require Data set $S$ and Query point $q$
\Ensure $\SphD(q;S)$
\item[]
\Procedure{Translating points}{}
\\input: $S$
\\output: Translated data set $S'$
\For{each $x_i \in S$ }
\State $x_i \gets (x_i-q)$ 
\EndFor
\State \Return $S'$
\EndProcedure
\item[]
\Procedure {Sorting around $T(q)$}{}
\\input: $T(q)$ and $S'$
\\output: Sorted array $S_{T}$
\For{each $x_i\in S'$}
\State Compute $\theta(x_i)$, where $\theta(x_i)$ is the angle of $x_i$ in polar coordinate system
\EndFor
\State Using an $O(n\log n)$ sorting algorithm, sort $x_i$ based on $\theta(x_i)$ in counterclockwise order
\State \Return $S_{T}$
\EndProcedure
\item[]  
\Procedure {Depth calculation}{}
\\input: $S_{T}$
\\output: Depth value of $\SphD(q;S)$
\State Initialize $\SphD(q;S)=0$
\For{each $x_{i}\in S_{T}$}

\State Initialize $a=0$ and $b=n+1$
\State Using two \emph{binary search calls}, update the values of $a$ and $b$
\item[]
\State \textbf{if} \; ($0< \theta(x_i) \leq \frac{\pi}{2}$)
\State \hspace{0.1cm} $a=min\{j:\: \theta(x_j)-\theta(x_i)\geq \frac{\pi}{2}\}$
\State \hspace{0.1cm} $b=max\{j:\: \frac{\pi}{2} \leq \theta(x_j)-\theta(x_i)\leq \frac{3\pi}{2}\}$
\State \textbf{else-if} \; ($\frac{\pi}{2}< \theta(x_i) \leq \frac{3\pi}{2}$)
\State \hspace{0.1cm} $a=max\{j:\: \theta(x_i)-\theta(x_j)\geq \frac{\pi}{2}\}$
\State \hspace{0.1cm} $b=min\{j:\: \theta(x_j)-\theta(x_i)\geq \frac{\pi}{2}\}$ 
\State \textbf{else}
\State \hspace{0.1cm} $a=min\{j:\: \theta(x_i)-\theta(x_j)\leq \frac{3\pi}{2}\}$
\State \hspace{0.1cm} $b=max\{j:\: \frac{\pi}{2} \leq \theta(x_i)-\theta(x_j)\leq \frac{3\pi}{2}\}$
\item[]
\State \textbf{if} \; $(a=0 \;\;\text{and
}\;\; b= n+1)$ 
\State $\hspace{0.1cm}|O_i|=0$
\State \textbf{else} \; Compute $\vert O_i\vert$
\State $\vert O_i\vert= \begin{cases}
a+(n-b+1)\; ; \frac{\pi}{2}< \theta(x_i) \leq \frac{3\pi}{2}\\
    b-a+1 \hspace{0.9cm}; \text{otherwise.}
\end{cases}
$
\item[]
\State $\SphD(q;S)\gets \SphD(q;S)+ \frac{1}{2} |O_i|$ 
\EndFor

\State \Return $\SphD(q;S)$
\EndProcedure 
\State \textbf{End}; 
\end{algorithmic}
\end{algorithm}

\begin{algorithm}
\renewcommand{\algorithmicrequire}{\textbf{Input:}}
\renewcommand{\algorithmicensure}{\textbf{Output:}}
\newcommand{\Break}{\State \textbf{break}}
\caption{Computing the $\beta$-skeleton depth of points in the plane}\label{Alg:betaskeleton-pseudocode}

\begin{algorithmic}[1]
\Require Data set $S$, Query point $q$, Parameter $\beta > 1$
\Ensure $\SkD_{\beta}(q;S)$
\item[]

\Procedure{Translating points}{}
\State Call Translating procedure from Algorithm~\ref{Alg:sph-pseudocode}. 
\EndProcedure
\item[] 
 
\Procedure {Depth calculation}{}
\\input: Translated data set $S'$
\\output: Depth value of $\SkD_{\beta}(q;S)$
\State Initialize $\SkD_{\beta}(q;S)=0$
\For{each $a \in S'$}
\State Using two $O(n^{\frac{1}{2}+\epsilon})$ range counting algorithms, compute $\vert H_o(p)\vert$ and $\vert intB_o(c,r) \vert$
\item[\hspace{1.5cm}$(H_o(p)$, $intB_o(c,r)$, $p$, $c$, and $r$ all are somefunctions can be computed in constant time)] 
\State \[\SkD_{\beta}(q;S)\gets \SkD_{\beta}(q;S)+\frac{1}{2}(\vert H_o(p)\vert -\vert intB^o(c,r)\vert )\] 
\EndFor

\State \Return $\SkD_{\beta}(q;S)$
\EndProcedure 
\State \textbf{End}; 
\end{algorithmic}
\end{algorithm}

\end{document}